\documentclass[conference,letterpaper]{IEEEtran}

\def\BibTeX{{\rm B\kern-.05em{\sc i\kern-.025em b}\kern-.08em
		T\kern-.1667em\lower.7ex\hbox{E}\kern-.125emX}}

\usepackage{microtype}
\usepackage{graphicx}
\usepackage{subcaption}
\usepackage{booktabs} 

\usepackage[utf8]{inputenc} 
\usepackage[T1]{fontenc}
\usepackage{url}
\usepackage{ifthen}
\usepackage[noadjust]{cite}
\usepackage[cmex10]{amsmath} 
\usepackage{amsthm,amssymb,amsfonts}
\usepackage{algorithm}
\usepackage[noend]{algpseudocode}
\usepackage{graphicx}
\usepackage{textcomp}
\usepackage{xcolor}
\usepackage{comment}
\usepackage{epstopdf}
\usepackage{float}
\usepackage{amssymb}
\usepackage{color}
\usepackage{relsize}
\usepackage{enumitem}
\usepackage{cleveref}
\usepackage{cases}
\usepackage[nocomma]{optidef}
\usepackage[mode=buildnew]{standalone}
\usepackage{breqn}
\usepackage{diagbox}

\usepackage{caption}
\setlist{leftmargin=4.1mm}

\theoremstyle{plain}
\newtheorem{theorem}{Theorem}
\newtheorem{lemma}{Lemma}
\newtheorem{definition}{Definition}

\newtheorem{corollary}[theorem]{Corollary}

\theoremstyle{definition}
\newtheorem{example}{Example}

\newcommand{\bfA}{\mathbf{A}}

\newcommand{\bfH}{\mathbf{H}}
\newcommand{\bfP}{\mathbf{P}}

\newcommand{\bfC}{\mathbf{C}}
\newcommand{\bfS}{\mathbf{S}}

\newcommand{\bfx}{\mathbf{x}}

\newcommand{\bfy}{\mathbf{y}}
\newcommand{\bfz}{\mathbf{z}}

\newcommand{\calB}{\mathcal{B}}

\newcommand{\calC}{\mathcal{C}}
\newcommand{\calI}{\mathcal{I}}
\newcommand{\calE}{\mathcal{E}}
\newcommand{\calS}{\mathcal{S}}

\newcommand{\F}{\mathbb{F}}

\newcommand{\nsubk}[2]{\binom{[#1]}{#2}}

\newcommand\lev[1]{{\color{black}#1}}

\begin{document}
	
	
	\title{Block-MDS QC-LDPC Codes for Information Reconciliation 
		in Key Distribution \vspace{-0.5cm}
	}
	
	
	\author{\IEEEauthorblockN{Lev Tauz, Debarnab Mitra, Jayanth Shreekumar,  Murat Can Sarihan, Chee Wei Wong, and Lara Dolecek}
		\IEEEauthorblockA{Department of Electrical and Computer Engineering, University of California, Los Angeles, USA\\
			email: \{levtauz, debarnabucla, jayshreekumar98, mcansarihan, cheewei.wong, and dolecek\}@ucla.edu
		}\vspace{-1.3cm}}
	
	\maketitle

	\begin{abstract}
		Quantum key distribution (QKD) is a popular protocol that provides information theoretically secure keys to multiple parties. Two important post-processing steps of QKD are 1) the information reconciliation (IR) step, where parties reconcile mismatches in generated keys through classical communication, and 2) the privacy amplification (PA) step, where parties distill their common key into a new secure key that the adversary has little to no information about. In general, these two steps have been abstracted as two distinct problems. In this work, we consider a new technique of performing the IR and PA steps jointly through sampling that relaxes the requirement on the IR step, allowing for more success in key creation. We provide a novel LDPC code construction known as Block-MDS QC-LDPC codes that can utilize the relaxed requirement by creating LDPC codes with pre-defined sub-matrices of full-rank. We demonstrate through simulations that our technique of sampling can provide notable gains in successfully creating secret keys. 
	\end{abstract}
	
	\vspace{-0.2cm}
	\section{Introduction and Motivation}
	\vspace{-0.1cm}
	
	\lev{Quantum communication technologies have already been identified as a valuable component of upcoming 6G systems for both communication and computation \cite{Wang2022, Nguyen2021}.} One important method \lev{in quantum communications} is Quantum Key Distribution (QKD) which allows for secret key agreement between two parties (Alice and Bob) using quantum mechanical principles to guarantee security against eavesdroppers (Eve) \cite{Zhong2015}. QKD is an important tool in a future where quantum computers are threatening to break many of the cryptographic protocols we rely on today and, thus, has received significant research attention \cite{Dolecek2022, Bradler2016,Zhuang2016,Zhang2018,Lee2019}.
	
	QKD can be broken down into 3 major steps: 1) Raw Key Generation: Alice and Bob generate keys from some quantum mechanical source and they have some measure about how much information Eve has about the keys; 2) Information Reconciliation (IR): Due to imperfections in the channel, Alice and Bob must reconcile the errors in their keys by communicating through a classical channel where Eve can eavesdrop; 3) Privacy Amplification (PA): Assuming the IR step was successful, Alice and Bob now distill their common key into a smaller key in order to remove any leaked information that Eve may have. In this paper, we study the interaction between the IR and PA steps in order to improve the overall performance of the QKD system. The main goal is to have a high \textit{secret key rate} which is the expected ratio of the final key length \lev{in bits} over the number of photons used to generate the keys. The secret key rate depends on the success probability of the IR step and the overall information provided to Eve. 
	
	To the best of our knowledge, many previous works have considered each step of the QKD process individually and have abstracted the problem into three separate problems \cite{Bennett1988,Elkouss2009,Mueller2023a}. In this work, we \lev{seek} to break the abstraction between the IR and PA steps in order to relax the requirements of the IR step, thereby allowing it to succeed more often and increase the secret key rate. The key idea of our work is that the PA step will be removing redundant information from the common key reconciled during the IR step. As such, it seems unnecessary for the IR protocol to reconcile all the mismatches if some are redundant and will be removed during the PA step anyway. By requiring the IR step to \textit{reconcile only a subset} of the key instead of the full key (essentially sampling the common key), we increase the probability that the IR step will succeed. This idea is similar in spirit to decoding of only the systematic bits in classical channel coding, which is known to provide \lev{significant} gains.
	
	
	Our contributions are as follows. First, we demonstrate an efficient privacy amplification technique through sampling that causes no information loss under certain practical conditions, thus relaxing the requirements for the IR step. Second, we construct a class of Quasi-Cyclic Low Density Parity Check (QC-LDPC) codes which we term as \textit{Block-MDS} QC-LDPC codes that work jointly with our privacy amplification technique. While designed with QKD in mind, we hypothesize that Block-MDS QC-LDPC codes can have further uses in other areas where LDPC codes are prominent. Finally, we provide simulation results to demonstrate the benefits of our joint IR/PA decoding technique. 
	
	The rest of this paper is organized as follows.  In Section \ref{sec:back}, we provide the preliminaries and the system model. In Section \ref{sec:samp}, we demonstrate our novel sampling technique for privacy amplification. In Section \ref{sec:mds}, we provide the design of our novel Block-MDS QC-LDPC codes.  Finally, we provide simulation results and concluding remarks in Section \ref{sec:sim}.
	
	\textit{Notation:} $\F_q$ denotes a finite field of order $q$. For positive integers $n$ and $m$, $\F_q^n$ ($\F_q^{n,m}$) denotes all vectors (matrices) of length $n$ (size $n\times m$) with elements from $\F_q$. For random variables $X$ and $Y$, $\calI(X;Y)$ denotes the mutual information between $X$ and $Y$ and $H(X)$ denotes the Shannon entropy of $X$. All logarithms are in base $2$. For positive integers $a$ and $b$, let $[a] ={1,2,\dots,a}$ and $(a)_b=a\mod b$. Given two integers $n$ and $k$ such that $k \leq n$, $\nsubk{n}{k}$ denotes all subsets of $[n]$ of size $k$. We shall denote all vectors by lowercase bold letters and matrices by uppercase bold letters. For a vector $\bfx$ (matrix $\bfH$) of size $n$ ($m\times n$) and set $\calS \subset [n]$, we denote $\bfx_{\calS}$ ($\bfH_{\calS}$) as the subset of the elements (columns) of $\bfx$ ($\bfH$) indexed by $\calS$. Let $S_n$ denote the set of all permutations of the set $[n]$.
	\vspace{-0.2cm}
	\section{Background and Model} \label{sec:back}
	\vspace{-0.1cm}
	%
	\subsection{System Model}
	\vspace{-0.1cm}
	As mentioned in the introduction, QKD systems can be broken down into 3 major components: Key Generation, Information Reconciliation, and Privacy Amplification. We shall describe each of these steps and focus on the relevant components of each step.
	
	1) \textit{Key Generation:} Alice and Bob generate raw keys using a quantum communication protocol such as an energy-time entanglement protocol \cite{Zhong2015,Chang2021}. Let $\bfx = \{x_1, \ldots, x_N\}$, $x_i\in \F_q$ and $\bfy = \{y_1, \ldots, y_N\}$, $y_i \in \F_q$ be the raw keys of length $N$ recorded by Alice and Bob, respectively. We assume that the random variables $x_i,i\in[N]$ are independent and uniform on $\F_q$. Due to imperfections in the detectors, the raw keys may differ in some positions. For simplicity, we assume that the symbol mismatch can be modeled by a $q$-ary symmetric channel where the errors are independent, see \cite{Mueller2023a}. As such, the conditional probability for $x_i$ given $y_i$ for $i\in [N]$ is 
	\begin{align}
		&Pr(x_i|y_i) = \begin{cases}
			1-p & y_i = x_i, \\
			\frac{p}{1-q} & \text{else},
		\end{cases} 
	\end{align}
	where $p$ denotes the channel transition probability. Additionally, the adversary Eve may contain some information about the raw keys which we denote as $\calE$.
	
	2) \textit{Information Reconciliation:} In this step, Alice and Bob reconcile the raw keys by communicating through a public channel which Eve has access to. Let $\bfz$ represent the data communicated between Alice and Bob which Eve can access. In this work, we consider single-round communication schemes which are equivalent to asymmetric Slepian-Wolf coding with side information at the receiver \cite{Elkouss2009}. We employ a linear coset scheme where Alice encodes the data $\bfx$ using a matrix $\bfH\in \F_q^{M,N}$ into syndrome $\bfz = \bfH\bfx$ and transmits $\bfz$ to Bob. Bob then uses the syndrome $\bfz$ and the side information $\bfy$ in order to decode $\bfx$. If Bob successfully decodes, then the protocol proceeds to the next step. If Bob fails to decode, then the algorithm stops and no key is generated. 
	
	3) \textit{Privacy Amplification:} In this step,  Alice and Bob start with a common key $\bfx$ since the IR step succeeded. Eve has information about $\bfx$ through $(\calE,\bfz)$ and Alice and Bob wish to distill $\bfx $ into a smaller key which is independent of $(\calE,\bfz)$. PA can be accomplished through the use of \textit{universal hash functions} \cite{Bennett1995}.  The length of the final key depends on the amount of information leaked from $(\calE,\bfz)$. Assuming that the PA step incurs no further information leakage, the final key length can be written as $H(\bfx)-\calI(\bfx;\calE,\bfz) = H(\bfx|\calE,\bfz)$.

	For a key distribution system, we consider the main measure of interest as the average number of generated bits in the final key per photon which is named the \textit{secret key rate}. Thus, the secret key rate can be defined as 
	\begin{align}
		SKR = Pr(A)\frac{H(\bfx)-\calI(\bfx;\calE,\bfz)}{N} = Pr(A)\frac{H(\bfx|\calE,\bfz)}{N}
	\end{align}
	where $A$ is the event that the IR step is successful.
	
	\subsection{LDPC code preliminaries}\label{subsection:ldpc}
	An LDPC code over $\F_q$ is defined by a sparse parity check matrix $\bfH \in \F_q^{M,N}$. For the coset scheme, LDPC codes can be decoded using a variant of the sum-product decoding algorithm specialized for the Slepian-Wolf problem (see \cite{Dupraz2015} for more details). All simulations in this work utilize this decoder.
	
	One method to construct an LDPC code is known as the scaled protograph-based method \cite{Thorpe2003, Dolecek2014}. This method starts with a small bipartite graph represented by a $\gamma\times\kappa$ base matrix of non-negative integers and the parity check matrix of the LDPC code is created by replacing each entry $a$ by a summation of $a$ scaled permutation matrices of size $z\times z$. We denote $\gamma$ as the column weight, $\kappa$ as the row weight, and $z$ as the lifting factor. When the base matrix is the all-ones matrix and the permutation matrices are all circulant shift matrices, then the resultant LDPC code is known as a Type-1 Quasi-Cycli LDPC (QC-LDPC) code \cite{Fossorier2004,Smarandache2012}. For the rest of this paper, we shall focus on these types of codes. Thus, the parity check matrix of QC-LDPC codes can be written as 
	\begin{equation}\label{eq:qc_ldpc}
		\bfH = \begin{bmatrix}
			s_{1,1}\bfC^{p_{1,1}} & s_{1,2}\bfC^{p_{1,2}} &\cdots & s_{1,\kappa}\bfC^{p_{1,\kappa}} \\
			s_{2,1}\bfC^{p_{2,1}} & s_{2,2}\bfC^{p_{2,2}} & \cdots & s_{2,\kappa}\bfC^{p_{2,\kappa}} \\			
			\vdots & & \ddots & \vdots\\
			s_{\gamma,1}\bfC^{p_{\gamma,1}} & s_{\gamma,2}\bfC^{p_{\gamma,2}} & \cdots & s_{\gamma,\kappa}\bfC^{p_{\gamma,\kappa}} \\			
		\end{bmatrix}
	\end{equation}
	where $\bfC^{p}$ is a circulant shift matrix (CSM) of size $z\times z$ with a one at column $r-p \mod z$ for row $r,0\leq r \leq z-1$ and zero elsewhere.
	We note that $\bfH$ can be uniquely determined by the scaling matrix $\bfS= \{s_{i,j}\}_{i\in[\gamma], j\in[\kappa]}, s_{i,j} \in \F_q$ and power matrix $\bfP = \{p_{i,j}\}_{i\in[\gamma], j\in[\kappa]}, 0 \leq p_{i,j}\leq z-1$.
	
	An important measure for LDPC codes is the girth, which is the length of the shortest cycle in the graph of the LDPC code. A necessary and sufficient condition for a QC-LDPC code to \lev{have a} certain girth is given in the following lemma:
	\begin{lemma}\cite{Fossorier2004} \label{lemma:fossorier}
		A QC-LDPC code in the form of Eq.\eqref{eq:qc_ldpc} has girth at least $2(g+1)$ if and only if
		\begin{equation}\label{eq:girth_cond}
			\sum_{k=1}^m p_{i_k,j_k} - p_{i_{k+1},j_k} \neq 0 \mod z
		\end{equation}
		for all $m$, $2 \leq m \leq g$, all $i_k$, $i\in[\gamma]$, and all $j_k$, $j\in[\kappa]$ with $i_1 = i_{m}$, $i_k \neq i_{k+1}$, and $j_k\neq j_{k+1}$.
	\end{lemma}
	
	Finally, we note that a matrix of size $m\times n$ with $m\leq n$ is considered Maximum-Distance Separable (MDS) if and only if every square submatrix of size $m\times m$ is full-rank.
	\section{Privacy Amplification with Sampling}  \label{sec:samp}
	
	In this section, we demonstrate how we can achieve privacy amplification by sampling the decoded sequence $\bfx$ under certain conditions. The benefit of this is that the IR decoder only needs to decode a certain subset of $\bfx$ which has a higher probability of success than fully decoding $\bfx$. We term the decoder that decodes the full $\bfx$ as the \textit{full codeword} (FC) decoder and the decoder that decodes a part of $\bfx$ as the \textit{subset codeword} (SC) decoder. We formally define the SC decoder as follows:
	
	\begin{definition}
		Given a set $\calS \subseteq [N]$, the SC decoder takes $\bfx_\calS$ from the IR step and inputs it into the PA step. As such, the secret key rate can be written as 
		\begin{align}
			SKR = Pr(\widetilde{A})\frac{H(\bfx_\calS)-\calI(\bfx_\calS;\calE,\bfz)}{N} 
		\end{align}
		where $\widetilde{A}$ is the event that $\bfx_{\calS}$ is decoded successfully in IR. 
	\end{definition}
	
	The following theorem provides sufficient conditions when the SC decoder \lev{cannot have a lower} secret key rate than the FC decoder. 
	
	\begin{theorem}\label{theorem:sampling_eve}
		Assume that there exists a set $\calS \subset [N], |\calS| = N-M$ such that the submatrix $\bfH_{\overline{\calS}}$ is full rank. Thus, we can write $\bfz = \bfH\bfx = \bfH_{\calS}\bfx_{\calS} + \bfH_{\overline{\calS}}\bfx_{\overline{\calS}}$. Additionally, assume that all the random variables $x_i,i\in[N]$ are conditionally independent given Eve's information $\calE$. \lev{If $SKR_1$ and $SKR_2$ are the secret key rates of the FC decoder and SC decoder, respectively, then $SKR_1 \leq SKR_2$.}
	\end{theorem}
	\begin{proof}
		First, we note that the probability of success for the FC decoder is clearly \lev{not higher} than the probability of success for the SC decoder since the event that $\bfx$ is correctly decoded is encompassed in the event that $\bfx_{\calS}$ is decoded. Thus, $Pr(\widetilde{A}) \geq Pr(A)$.
		Next, we note that 
		\begin{equation}\label{eq_skr1}
			\begin{split}
				&H(\bfx)-\calI(\bfx;\calE,\bfz) \stackrel{(a)}{=} H(\bfx) - \calI(\bfx;\calE) - \calI(\bfx;\bfz| \calE)  \\
				&= H(\bfx) - \left(H(\bfx) - H(\bfx| \calE)\right) - \left(H(\bfz|\calE) - H(\bfz| \bfx,\calE)\right) \\ 
				&= H(\bfx| \calE) - H(\bfz|\calE) + H(\bfz| \bfx,\calE) \\
				&\stackrel{(b)}{=} H(\bfx| \calE) - H(\bfz|\calE)
			\end{split}	 
		\end{equation}
		where (a) uses the chain rule for mutual information and (b) uses the fact that $H(\bfz| \bfx,\calE) = 0$ due to $\bfz$ being a deterministic function of $\bfx$. We can use a similar logic for the following:
		\begin{equation}
			H(\bfx_{\calS})-\calI(\bfx_{\calS};\calE,\bfz) =H(\bfx_{\calS}| \calE) - H(\bfz|\calE) + H(\bfz| \bfx_{\calS},\calE) .
		\end{equation}
		We note that 
		\begin{align}
			H(\bfz| \bfx_{\calS},\calE) &= H(\bfH_{\calS}\bfx_{\calS} + \bfH_{\overline{\calS}}\bfx_{\overline{\calS}}| \bfx_{\calS},\calE) \stackrel{(a)}{=}  H(\bfH_{\overline{\calS}}\bfx_{\overline{\calS}}| \bfx_{\calS},\calE) \nonumber \\
			&\stackrel{(b)}{=} H(\bfH_{\overline{\calS}}\bfx_{\overline{\calS}}| \calE)
			\stackrel{(c)}{=} H(\bfx_{\overline{\calS}}|\calE)
		\end{align}
		where $(a)$ arises from removing the contribution of $\bfx_{\calS}$ in $\bfz$, (b) comes from the conditional independence of the r.v. in $\bfx$ when conditioned on $\calE$, and (c) comes from the fact that $\bfH_{\overline{\calS}}$ is a square full rank matrix and, thus, a bijective operation that preserves entropy. Thus, we have 
		\begin{align}
			&H(\bfx_{\calS})-\calI(\bfx_{\calS};\calE,\bfz) = H(\bfx_{\calS}| \calE) - H(\bfz|\calE) + H(\bfx_{\overline{\calS}}|\calE) \nonumber \\
			&\stackrel{(a)}{=} H(\bfx| \calE) - H(\bfz|\calE) \stackrel{(b)}{=} H(\bfx)-\calI(\bfx;\calE,\bfz)
		\end{align}
		where (a) arises from the conditional independence of $\bfx$ when conditioned on $\calE$ which results in $H(\bfx| \calE) = H(\bfx_{\calS}| \calE) + H(\bfx_{\overline{\calS}}|\calE)$ and (b) comes from Eq. \eqref{eq_skr1}.
		
		Thus, we have proven that the final key lengths are the same and that the probability of success of the SC decoder is \lev{not lower than for the FC decoder} which guarantees $SKR_1 \leq SKR_2$.
	\end{proof}

	The key idea of Theorem \ref{theorem:sampling_eve} is that carefully sampling $\bfx$ allows us to use the entropy of the leftover bits to increase privacy despite the reconciled vector $\bfx_\calS$ being smaller. In total, the final key length is the same for both decoders. The proposed approach relaxes the success condition for the IR step. Additionally, the proof of Theorem \ref{theorem:sampling_eve} did not rely on $\calS$ being the only set with this property. We can thus generalize the SC decoder to decoding at least one of multiple subsets with the full rank property. The following definition provides a description of this decoder:
	\begin{definition}
		Let $\mathbb{S} = \{\calS_i: i \in[k]\}$ be a set of $k$  subsets of $[n]$ that are possibly non-disjoint. The \textit{multiple subset codeword} (MSC) decoder samples the subset $\bfx_{\calS_i}$ with the highest secret key rate as defined by
		\begin{align}
			SKR_i = Pr(\widetilde{A}_i)\frac{H(\bfx_X)-\calI(\bfx_X;\calE,\bfz)}{N} , i\in [k]
		\end{align}
		where $\widetilde{A}_i$ is the event that $\bfx_{\calS_i}$ is decoded successfully in IR. 
	\end{definition}
	
	Thus, we get the following corollary of Theorem \ref{theorem:sampling_eve} for the MSC decoder.
	\begin{corollary}\label{corollary:sampling}
		If $|\calS| = N-M$ and $\bfH_{\overline{\calS}}$ is full rank for every $\calS \in \mathbb{S}$, then the MSC decoder achieves a secret key rate \lev{that is equal to or greater than the secret key rate of an} SC decoder for any particular $\calS \in \mathbb{S}$.
	\end{corollary}
	
	In the sequel, we assume that $\mathbb{S}$ satisfies Corollary \ref{corollary:sampling} whenever we discuss the MSC decoder. We note that the MSC decoder works naturally with any probability-based decoder, such as the belief propagation decoder of LDPC codes that can output a subset with the highest probability of being correct. In the next section, we demonstrate how to construct codes that can be utilize the MSC decoder.

	\section{Block-MDS QC-LDPC Codes} \label{sec:mds}
	
	In this section, we demonstrate how to construct QC-LDPC codes for the MSC decoder. In theory, we could randomly sample an LDPC code from a code ensemble and find all the square full rank submatrices of the parity check matrix. Yet, this approach would be quite difficult to analyze since the number of full rank submatrices can differ between samples. As such, we turn towards structured codes such as QC-LDPC codes and devise construction methods that guarantee certain subsets have the full rank property. We formally define this notion as follows: 
	
	\begin{definition}\label{def:block_mds}
		A QC-LDPC code is \textbf{Block-MDS} if all the sub-matrices $\bfH_{\calS_{\calB}}, \calB  \in \nsubk{\kappa}{\gamma}$  where $\calS_{\calB} \triangleq  \{(i-1)\times z+(j-1) : i\in \calB, j\in [z]\}$ where $\kappa$ is the row weight, $\gamma$ is the column weight, and $z$ is the lifting factor.
	\end{definition}
	
	At a high level, a Block-MDS QC-LDPC code guarantees that every square submatrix that corresponds to the lifting of a $\gamma \times \gamma$ submatrix in the parity check matrix of the protograph is full-rank. This is conceptually similar to an MDS matrix where every square submatrix is full rank but instead we focus on the lifted block matrices being full rank. As such, the MSC decoder subsets for the Block-MDS code are $\mathbb{S} = \{\overline{\calS_{\calB}}: \calB  \in \nsubk{\kappa}{\gamma}\}$. Example \ref{example:block_mds} demonstrates Definition \ref{def:block_mds}.
	
	
	\begin{example}\label{example:block_mds}
		Consider the following parity check matrix of a QC-LDPC code with $(\gamma,\kappa) = (2,3)$ (see Section \ref{subsection:ldpc}):
		\begin{equation}
			\bfH = \begin{bmatrix}
				s_{1,1}\bfC^{p_{1,1}} & s_{1,2}\bfC^{p_{1,2}} & s_{1,3}\bfC^{p_{1,3}} \\
				s_{2,1}\bfC^{p_{2,1}} & s_{2,2}\bfC^{p_{2,2}} & s_{2,3}\bfC^{p_{2,3}} \\			
			\end{bmatrix}.
		\end{equation}
		$\bfH$ is Block-MDS if the following submatrices are full rank
		\begin{align*}
			&\bfH_{\calS_{1,2}} = \begin{bmatrix}
				s_{1,1}\bfC^{p_{1,1}} & s_{1,2}\bfC^{p_{1,2}}  \\
				s_{2,1}\bfC^{p_{2,1}} & s_{2,2}\bfC^{p_{2,2}}\\			
			\end{bmatrix}, \\ 
			&\bfH_{\calS_{1,3}} = \begin{bmatrix}
				s_{1,1}\bfC^{p_{1,1}} & s_{1,3}\bfC^{p_{1,3}} \\
				s_{2,1}\bfC^{p_{2,1}} & s_{2,3}\bfC^{p_{2,3}} \\			
			\end{bmatrix},\\
			&\bfH_{\calS_{2,3}} = \begin{bmatrix}
				s_{1,2}\bfC^{p_{1,2}} & s_{1,3}\bfC^{p_{1,3}} \\
				s_{2,2}\bfC^{p_{2,2}} & s_{2,3}\bfC^{p_{2,3}} \\	
			\end{bmatrix}.
		\end{align*}
	\end{example}
	
	By focusing on Block-MDS QC-LDPC codes, we can significantly simplify the design of LDPC codes that can utilize the MSC decoder. For the rest of this section, we shall investigate techniques to construct Block-MDS QC-LDPC codes. We first state an important result in linear algebra that we rely on extensively in this paper:
	
	\begin{lemma}{\cite[Theorem 1]{Silvester2000}} \label{lemma:det}
		Let $\mathcal{R}$ be a commutative subring of $\F_q^{z, z}$, i.e., $\mathcal{R}$ is a set of matrices of size $z\times z$ that form a commutative ring with the standard operations of matrix addition and multiplication. Let $\mathbf{M} \in \mathcal{R}^{a\times b}$, i.e. $\mathbf{M}$ is a block matrix where each block is an element in $\mathcal{R}$. Then, 
		\begin{equation}
			\det_{\F_q}(\mathbf{M}) = \det_{\F_q}(\det_{\mathcal{R}}(\mathbf{M})),
		\end{equation}
		where $\det_{F}$ is the determinant function over a ring $F$.
	\end{lemma}
	
	Consider the set $\calC \subset \F_q^{z,z}$ as the set of all circulant matrices of size $z\times z$ with elements in the field $\F_q$. It is well known that $\calC$ is a commutative ring in regards to operations of the standard matrix addition and multiplication \cite[Theorem 7.3.2]{Hachenberger2020a}. Since a QC-LDPC code is a block matrix consisting of CSMs, Lemma \ref{lemma:det} states that a necessary and sufficient condition for the QC-LDPC code to be Block-MDS is that it satisfies 
	\begin{equation}
		\det_{\F_q}\left(\sum_{\sigma \in S_{\gamma}}  sign(\sigma) \prod^{\gamma}_{i=1} s_{\sigma(i),\tau(i)}\bfC^{p_{\sigma(i),\tau(i)}}\right) \neq 0, \;  \forall \tau \in \nsubk{\kappa}{\gamma},\label{eq:det_condition}
	\end{equation}
	where we have expressed the determinant function using the well-known Leibniz formula and $sign(\sigma)$ is the parity of the permutation $\sigma$. Note that the inner sum must be a circulant due to $\calC$ being a commutative ring. Thus, the Block-MDS condition can be checked for a particular QC-LDPC code by whether $\binom{\kappa}{\gamma}$ circulant matrices of size $z\times z$ are singular. The direct way would be to take the determinant of each circulant matrix in the field $\F_q$. For circulant matrices, there is a much easier check for singularity. First, let us define the associated polynomial of a circulant matrix as $f(x) = \sum_{i=0}^{z-1}a_ix^i$ where $a_i$ is the $i^{\text{th}}$ element in the first column of the circulant matrix. The following lemma provides a simple condition to check whether a circulant matrix is singular \cite{Ingleton1956,Fabsic2017}:
	
	\begin{lemma}\label{lemma:circ_root}
		Let $f(x)$ be the associated polynomial of a circulant matrix $\bfA \in \F_q^{z,z}$. Then, $\bfA$ is non-singular if and only if $\gcd(f(x),x^z-1) = 1 $.
	\end{lemma}

	Using Lemmas \ref{lemma:det} and \ref{lemma:circ_root}, we arrive at the following theorem: 
	\begin{theorem}\label{theorem:sufficient_cond}
		A sufficient condition for a QC-LDPC code with parameters $(\gamma,\kappa,z)$ to be Block-MDS is that the scaling matrix $\bfS$ and power matrix $\bfP$ satisfy

		\begin{equation}\label{eq:suff_1a}
			\gcd(f_{\tau}(x), x^z-1) = 1,
		\end{equation}
		\begin{equation}\label{eq:suff_1b}
			f_{\tau}(x) = \sum_{\sigma \in S_{\gamma}}sign(\sigma)\left(\prod^{\gamma}_{i=1} s_{\sigma(i),\tau(i)}\right) x^{\left(\sum^{\gamma}_{i=1}p_{\sigma(i),\tau(i)}\right)_z},
		\end{equation}
		\begin{equation}\label{eq:suff_2}
			\sum^{\gamma}_{i=1}p_{\sigma(i),\tau(i)} \neq \sum^{\gamma}_{i=1}p_{\rho(i),\tau(i)},  \;  \forall  \rho,\sigma \in S_{\gamma}, \rho\neq\sigma,
		\end{equation}
		for all $\tau \in \nsubk{\kappa}{\gamma}$.
	\end{theorem}
	
	\begin{proof} 
		To simplify Eq. \eqref{eq:det_condition}, we can enforce that all circulant matrices in the inner sum (after performing the products) do not have any overlap in their non-zero positions. This ensures that each matrix contributes to only one coefficient in the associated polynomial of the summed up circulant matrix. 
		Eq. \eqref{eq:suff_2} accomplishes this by requiring that for a given $\tau$ all the matrix powers in that particular sum are distinct which ensures no overlap in the non-zero terms of the summed circulant matrix. 
		As such, the associated polynomial $f_\tau(x)$ for a given $\tau$ can be written as 
		Eq. \eqref{eq:suff_1b}. Applying Lemma \ref{lemma:circ_root} results in Eq. \eqref{eq:suff_1a} which completes the proof.
	\end{proof}

	At first glance, Theorem \ref{theorem:sufficient_cond} seems to provide a sufficient condition that is quite restrictive on the parameters due to Eq.\eqref{eq:suff_2}. In fact, the following example demonstrates that Theorem \ref{theorem:sufficient_cond} \lev{broadly} applies to QC-LDPC codes of high girth which are attractive for their error correcting performance. 
	\begin{example}\label{example:cycle_cond}
		Consider the QC-LDPC code in Example \ref{example:block_mds}. According to Theorem \ref{theorem:sufficient_cond}, the following equations are sufficient for this QC-LDPC code to be Block-MDS:
		\begin{align}
			&\gcd(s_{1,1}s_{2,2}x^{(p_{1,1}+p_{2,2})_z}-s_{2,1}s_{1,2}x^{(p_{2,1}+p_{1,2})_z}, x^z-1) = 1 \\
			&\gcd(s_{1,1}s_{2,3}x^{(p_{1,1}+p_{2,3})_z}-s_{2,1}s_{1,3}x^{(p_{2,1}+p_{1,3})_z}, x^z-1) = 1 \\
			&\gcd(s_{1,2}s_{2,3}x^{(p_{1,2}+p_{2,3})_z}-s_{2,2}s_{1,3}x^{(p_{2,2}+p_{1,3})_z}, x^z-1) = 1 \\
			&p_{1,2}+p_{2,3} \neq p_{2,2}+p_{1,3} \mod z \label{eq:c1}\\
			&p_{1,1}+p_{2,3} \neq p_{2,1}+p_{1,3} \mod z\label{eq:c2} \\
			&p_{1,2}+p_{2,3} \neq p_{2,2}+p_{1,3} \mod z\label{eq:c3}
		\end{align}
	\end{example}
	
	Note that Eqs.\eqref{eq:c1},\eqref{eq:c2},\eqref{eq:c3} are a subset of the cycle conditions in Lemma \ref{lemma:fossorier} to ensure that the QC-LDPC code has no cycles of length 4. In fact, we can see that Eq. \eqref{eq:suff_2} in Theorem \ref{theorem:sufficient_cond} is always a subset of the cycle conditions in Lemma \ref{lemma:fossorier} for containing no cycles of length $\gamma$. Thus, we  get the following corollary:
	\begin{corollary}\label{corollary:girth}
		A QC-LDPC code with column weight $\gamma$ and girth $2\gamma+2$ is Block-MDS if and only if it satisfies the equations in Theorem \ref{theorem:sufficient_cond}.
	\end{corollary}
	
	Thus, Theorem \ref{theorem:sufficient_cond} is sufficient to guarantee Block-MDS among high girth QC-LDPC codes which are the class of QC-LDPC codes that we generally focus on due to their higher error-correcting performance. We note that Corollary \ref{corollary:girth} becomes less meaningful for $\gamma \geq 6$ as it is well known that type-I QC-LDPC codes have a minimum girth of 12 \cite{Fossorier2004}. This is not a problematic constraint since many practical type-I QC-LDPC codes generally have $\gamma$ be 3 or 4. A future research direction is generalizing our result to more complex constructions of QC-LDPC codes that permit a higher girth.
	
	
	For special values of the lifting factor $z$, Theorem \ref{theorem:sufficient_cond} can also be used to derive a simpler condition that allows for decoupling the search for \lev{matrices} $\bfS$ and $\bfP$. The following theorem provides sufficient conditions where a high girth QC-LDPC code can be made into a Block-MDS code where the finite field size scales linearly with $\kappa$.
	
	\begin{theorem}\label{theorem:vand}
		If the lifting factor $z$ is an odd prime and the function $\sum_{i=0}^{z-1}x^i$ is irreducible in $\F_q$, then a QC-LDPC code with girth $2\gamma+2$ can be made into a Block-MDS code with a careful choice of $\bfS$ for all $\kappa \leq |\F_q|$ and $\gamma! < z$. 
	\end{theorem}
	
	\begin{proof}
		Let us consider Eq.\eqref{eq:suff_1a}. When $z$ is a prime, then we can easily factor $x^z-1$ into $(x-1)(\sum_{i=0}^{z-1}x^i)$. By the theorem statement, these are the irreducible factors of $x^z-1$. The left factor indicates that for the gcd to be $1$, then $1$ cannot be a root of $f_\tau(x)$, i.e., 
		\begin{equation}
			f_{\tau}(1) = \sum_{\sigma \in S_{\gamma}}sign(\sigma)\left(\prod^{\gamma}_{i=1} s_{\sigma(i),\tau(i)}\right) \neq 0 \lev{\; \in \F_q}.
		\end{equation}
		Note that $f_{\tau}(1)$ is simply the determinant of the $\gamma \times \gamma$ submatrix of $\bfS$ where the columns are selected by $\tau$. Since this \lev{condition} needs to be true for every choice of $\tau$, then $\bfS$ must be an MDS matrix. Now, we only need to prove that $f_{\tau}(x)$ is not a factor of $\sum_{i=0}^{z-1}x^i$ since the degree of $f_{\tau}(x)$ is less than or equal to $z-1$. Since $\sum_{i=0}^{z-1}x^i$ is irreducible, we only need to show that $\sum_{i=0}^{z-1}x^i \neq f_{\tau}(x)$. This is true by noting that the number of non-zero elements in the polynomial $f_{\tau}(x)$ is upper bounded by $\gamma!$ which is less than $z$ by the theorem statement. Hence, Eq.\eqref{eq:suff_1a} is equivalent to requiring that $\bfS$ is an MDS matrix. 
		
		We complete the proof by using the well-known Vandermonde matrix of size $\gamma \times \kappa$ for $\bfS$ since it is MDS and it only needs a field size of $\kappa\leq |\F_q|$ \cite{Klinger1967}.
	\end{proof}
	Theorem \ref{theorem:vand} allows us to decouple the constructions of \lev{matrices} $\bfP$ and $\bfS$. Thus, we can first find a \lev{matrix} $\bfP$ with sufficient girth properties and then transform it using an easily defined \lev{matrix} $\bfS$ where the finite field size scales linearly with the row weight. This \lev{property} is very useful in practice since large finite field sizes incur significant complexity in decoding which translates to higher latency or more complex circuitry. Our design allows for Block-MDS QC-LDPC codes that are almost independent of the block length since the field size depends on $\kappa$ for lifting factors that satisfy Theorem \ref{theorem:vand}.
	
	
	
	\begin{table}[]
		\caption{Parameters for Codes used in Simulations. All lifting factors $z$ were chosen to acquire codes close to length $2000$ for fair comparison while satisfying the conditions in Theorem \ref{theorem:vand}.}\label{table:params}
		\vspace{-0.1cm}
		\begin{center}
			\begin{tabular}{l|cccc}
				Code & $(\gamma,\kappa)$      & Lifting Factor & Rate & Length \\ \hline
				$C_1$ & (3,4) & 491            & 1/4  & 1964   \\
				$C_2$ & (3,5) & 389            & 2/5  & 1945   \\
				$C_3$ & (4,5) & 389            & 1/5  & 1945  
			\end{tabular}
			\vspace{-0.75cm}
		\end{center}
	\end{table}
	
	\begin{table}[]
		\caption{Secret Key Rates at representative points for high noise regime.}\label{table:skr}
		\vspace{-0.1cm}
		\begin{center}
			\begin{tabular}{l|ccc}
				Code & Transition Probability & FC SKR     & MSC SKR  \\ \hline
				$C_1$ & $p=0.275$ & 0.3913  & 0.4832  \\
				$C_2$ & $p=0.2$  & 0.8883  & 0.9679  \\
				$C_3$ &  $p=0.28$ & 0.4114          & 0.45  
			\end{tabular}
			\vspace{-0.5cm}
		\end{center}
	\end{table}
	\begin{figure}[t]
		\centering
		\includegraphics[width=0.6\linewidth]{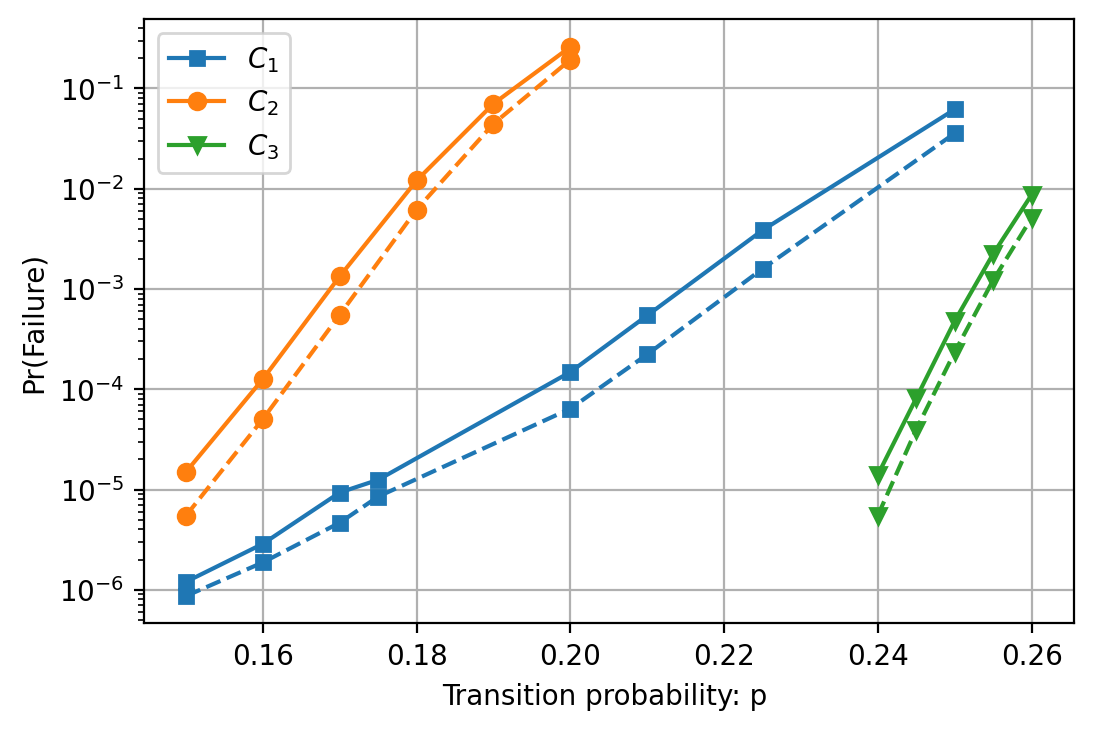}
		\caption{Probability of IR failure for different transition probabilities for a 8-ary symmetric channel. Bold line indicates the FC decoder and dotted lines indicates the MSC decoder. }\label{fig:sims}
		\vspace{-0.6cm}
	\end{figure}

	\section{Simulations and Conclusion} \label{sec:sim}
	\vspace{-0.2cm}
	
	In this section, we shall demonstrate the benefits of using our new decoding method to jointly perform information reconciliation and privacy amplification on our Block-MDS QC-LDPC codes. We shall be comparing the secret key rate using FC and MSC decoding on our Block-MDS QC-LDPC codes to demonstrate the gains offered by the relaxation of the IR step.
	Since the final key length for a code is the same regardless of the decoder chosen (FC or MSC), the major measure of interest is the IR failure probability for the secret key rate. As such, we shall demonstrate the improvements that the MSC decoder has over the FC decoder in terms of the IR failure probability for the low noise regime and the secret key rate at the high noise regime.

	We perform simulations on $3$ QC-LDPC codes with parameters described in Table \ref{table:params}. 
	All codes were constructed to have girth $10$. The power matrix $\bfP$ and scaling matrix $\bfS$ for each code can be found in Appendix \ref{app:code}. Fig. \ref{fig:sims} plots the probability of IR failure for different values of the transition probability for an $8$-ary symmetric channel. We see that the MSC decoder can improve the  IR failure probability by about $0.25$ orders of magnitude. Clearly, the gains differ for different code parameters which suggests further study into how code parameters affect the decoding probability of the MSC decoder.  Yet, we can say that the MSC decoder can provide significant gains. Additionally, Table \ref{table:skr} demonstrates the improvement in the secret key rate at the high noise regime which is commonly found in practice. In this regime, even a small improvement in the FER can have significant gains in the secret key rate as demonstrated by the MSC decoder. 
	
	
	In conclusion, we have demonstrated a powerful relaxation for the IR step in QKD, thus allowing us to improve the success rate of the IR step. This relaxation comes from a novel sampling technique between the IR and PA step. Additionally, we provide a novel LDPC code design in the form of Block-MDS QC-LDPC codes that can capitalize on this relaxation. We empirically demonstrate the improvements of our new decoder on these LDPC codes through simulations. Future work is focused on generalizing our ideas to \lev{a broader set of graph codes}.

	\newpage
	\bibliographystyle{ieeetr}
	\bibliography{references_qkd}
	
	
	\newpage
	\appendix 
	
	\subsection{Code Parameters} \label{app:code}
	Since all elements of $\bfS$ are in $\F_8$, we provide the binary representation of each element for $\bfS$.
	
	Code 1:
	\begin{align}
		&\bfP = \begin{bmatrix}
			0&0&0&0 \\
			0&1&11&26 \\
			0&18&4&6  \\
		\end{bmatrix} \\
		&\bfS = \begin{bmatrix}
			1&1&1&1 \\
			1&2&3&4 \\
			1&4&5&6 \\
		\end{bmatrix}
	\end{align}
	Code 2:
	\begin{align}
		&\bfP = \begin{bmatrix}
			0&0&0&0&0 \\
			0&1&13&3&24 \\
			0&37&75&22&8 \\
		\end{bmatrix} \\
		&\bfS = \begin{bmatrix}
			1&1&1&1&1 \\
			1&2&3&4&5 \\
			1&4&5&6&7 \\
		\end{bmatrix}
	\end{align}
	Code 3:
	\begin{align}
		&\bfP = \begin{bmatrix}
			0&0&0&0&0 \\
			0&9&2&29&76 \\
			0&120&19&6&161 \\
			0&43&109&158&12 \\
		\end{bmatrix} \\
		&\bfS = \begin{bmatrix}
			1&1&1&1&1 \\
			1&2&3&4&5 \\
			1&4&5&6&7 \\
			1&3&4&5&6 \\
		\end{bmatrix}
	\end{align}
	
	
	
	

\end{document}